\documentclass[9pt,twocolumn]{extarticle}

\usepackage{soul}
\usepackage{titlesec} 
\usepackage[backend=bibtex,style=nature]{biblatex}
\DeclareNameAlias{sortname}{first-last}
\DeclareNameAlias{default}{first-last}
\bibliography{biblio} 
\usepackage{dsfont}
\usepackage{graphicx}
\usepackage{subcaption}
\usepackage[margin=0.9in]{geometry}
\usepackage[usenames,dvipsnames]{color}
\usepackage[colorlinks,linkcolor=Blue,urlcolor=Blue,citecolor=Blue]{hyperref}

\usepackage{amsmath,amssymb,braket}

\usepackage{setspace}

\usepackage{amsfonts}
\usepackage{bbm}
\usepackage[squaren]{SIunits}
\usepackage{cuted} 
\usepackage{array}
\newcolumntype{P}[1]{>{\arraybackslash}p{#1}}
\newcolumntype{Q}[1]{>{\centering\arraybackslash}p{#1}}

\usepackage{diagbox}

\usepackage{mathpazo}

\usepackage{courier}
\normalfont

\usepackage[T1]{fontenc}
\usepackage{caption}  
\renewcommand{\thefigure}{\arabic{figure}}

\usepackage{upgreek}

\newcommand{\ad}[1]{\textsuperscript{#1}\kern-2pt}

\usepackage{amsmath,amsthm,mathtools}

\makeatletter
\def\blx@maxline{77}
\makeatother

\usepackage{indentfirst}
\usepackage{capt-of}

 \usepackage{flushend}

\usepackage[ruled]{algorithm2e} 
\usepackage{newfloat,algcompatible} 
\usepackage{tikz}

\def\mytitle{
Simulating Hamiltonian dynamics in a programmable photonic quantum processor using linear combinations of unitary operations\vspace{-4mm}}   

\title{\vspace{-1cm}\huge\textbf{\textrm{\mytitle}}}  

\author{Yue Yu$^{1,\star}$, Yulin Chi$^{1,\star}$, Chonghao Zhai$^{1}$, Jieshan Huang$^{1}$,  Qihuang Gong$^{1,2,3,4,5}$,  Jianwei Wang$^{1,2,3,4,5,\dagger}$}

\date{} 
\begin{document}
\twocolumn[{
\maketitle 
\vspace{-9mm}
\begin{center}
\begin{minipage}{1\textwidth}
\begin{center}
\textit{\textrm{
\textsuperscript{1} State Key Laboratory for Mesoscopic Physics, School of Physics, Peking University, Beijing, 100871, China
\\\textsuperscript{2} Frontiers Science Center for Nano-optoelectronics \& Collaborative Innovation Center of Quantum Matter, Peking University, Beijing, 100871, China 
\\\textsuperscript{3} Collaborative Innovation Center of Extreme Optics, Shanxi University, Taiyuan 030006, Shanxi, China
\\\textsuperscript{4} Peking University Yangtze Delta Institute of Optoelectronics, Nantong 226010, Jiangsu, China.
\\\textsuperscript{5} Hefei National Laboratory, Hefei 230088, China\\
{$\star$} These authors contributed equally to this work. ~~~{$\dagger$}  emails to:  jww@pku.edu.cn
  }}
\end{center}
\end{minipage}
\end{center}

\setlength\parindent{12pt}
\begin{quotation}
\noindent 
{Simulating the dynamic evolutions of physical  and molecular systems in a quantum computer is of fundamental interest in many applications~\supercite{RevModPhys.86.153,mcardle2020quantum}. Its implementation requires efficient quantum   simulation algorithms~\supercite{trotter1959product,suzuki1976generalized,lloyd1996universal,aspuru2005simulated,berry2015simulating, berry2017exponential,low2019Hamiltonian,low2017optimal,campbell2019random,wiebe2011simulating}. 
The Lie-Trotter-Suzuki approximation algorithm, also well known as the Trotterization, is a basic algorithm in quantum dynamic simulation\supercite{trotter1959product,suzuki1976generalized,lloyd1996universal}. A multi-product algorithm that is a linear combination  of multiple Trotterizations  has  been proposed  to improve the approximation accuracy~\supercite{childs2012hamiltonian}. 
Implementing such multi-product Trotterization in quantum computers however remains experimentally challenging and its success probability is limited. 
Here, we modify the multi-product Trotterization  and combine  it with the oblivious amplitude amplification to simultaneously  reach a high simulation precision and high success probability. 
We experimentally implement the modified multi-product algorithm in an integrated-photonics programmable quantum simulator in silicon, which allows the initialization, manipulation and measurement of four-qubit states and a sequence of linearly combined controlled-unitary gates, 
to emulate  the dynamics of a coupled electron and nuclear spins  system. 
Theoretical and experimental results are in good agreement, and they both show the modified multi-product algorithm can simulate Hamiltonian dynamics with a higher precision  than conventional Trotterizations and a nearly deterministic success probability. 
We certificate the multi-product algorithm in a small-scale quantum simulator based on linear combinations of operations, and this work promises the practical  implementations of quantum dynamics simulations.}
\end{quotation}}]

\noindent Efficiently simulating the quantum dynamics of physical and molecular systems represents an  important near-term application of quantum computers~\supercite{RevModPhys.86.153,mcardle2020quantum}. 
Several quantum algorithms have been proposed for simulating quantum Hamiltonian dynamics, e.g., Lie-Trotter-Suzuki approximations~\supercite{trotter1959product,suzuki1976generalized,lloyd1996universal} (known as Trotterizations), quantum random walks~\supercite{berry2009black, childs2010relationship}, multi-product formula~\supercite{childs2012hamiltonian}, truncated Taylor series~\supercite{berry2015simulating, berry2017exponential},   qubitization~\supercite{low2019Hamiltonian} and quantum signal processing~\supercite{low2017optimal}. 
Among them, Trotterization is regarded as a basic algorithm in  dynamics simulations.  The  Trotterization accuracy  however appears to be orders of magnitude looser than its prediction error bounded by numerical simulations
~\supercite{babbush2015chemical,heyl2019quantum}. Many algorithms have been proposed to improve  the Trotterizations, e.g, product order randomization~\supercite{campbell2019random}, time-dependent Hamiltonian simulation~\supercite{wiebe2011simulating}, and truncated Taylor series~\supercite{berry2015simulating, berry2017exponential}.
The order of  Lie-Trotter-Suzuki approximation determines the error of Trotterization. For example, the first-order formula requires a $\mathcal{O}(t^2/\epsilon)$-step Trotterization with an error of  $\epsilon$. Expectably, utilizing  higher-order formula can result in a substantial improvement of precision at the expense of more quantum operations and higher circuit depth~\supercite{sornborger1999higher,berry2007efficient}. 
Another efficient multi-product algorithm was proposed by \textit{Childs} and \textit{Wiebe}~\supercite{childs2012hamiltonian}, which can reach the same precision as the high-order Trotterization but with a lower circuit depth.

The quantum hardware for implementing the multi-product algorithm is a linear combinations of unitaries (LCU)\supercite{childs2012hamiltonian}. Note the LCU enables a general form of quantum operations that can be used to implement Hamiltonian simulation\supercite{childs2012hamiltonian,berry2015simulating}, Harrow-Hassidim-Lloyd algorithm\supercite{wei2017realization}, passive quantum error correction\supercite{marshman2018passive} and simulation of the Yang-Baxter equation\supercite{zheng2018duality}.
The realization of LCU circuits in quantum devices requires the implementation of a sequence of multi-qubit controlled-unitary operations~\supercite{gui2006general}. 
Realizing such multi-qubit controlled-unitary operations, in which the unitary can be arbitrarily  controllable, is generally a challenging task, though multi-qubit gates have been reported in different quantum systems, e.g, superconducting qubits \supercite{Huang2020}, trapped ions \supercite{Monroe,} and photons \supercite{Geoff}. 
In photonic systems, controlled-unitary operations between qubits or qudits have been demonstrated
\supercite{Zhou2011,Wang:QHL,Santagati2018,chi2022programmable,Patel}. 
Silicon-photonics quantum technologies \supercite{Wangreview}, that can integrate entangled-photon sources \supercite{vigliar2021error}, programmable quantum circuits \supercite{Wang16D,qiang2018large}, and single-photon detectors \supercite{pernice2012high}, could provide a versatile platform for multi-qubit  LCU-based Hamiltonian dynamics simulation.

Here, we propose and demonstrate a modified multi-product algorithm, assisted by  the oblivious amplitude amplification (OAA).  The  modified multi-product algorithm allows a higher accuracy of dynamics simulation compared to  conventional Trotterizations as well as the original multi-product \supercite{childs2012hamiltonian}, 
and it also enables a nearly deterministic success probability. 
We experimentally implement our modified multi-product algorithms in a LCU quantum simulator on a programmable four-qubit  quantum photonic chip. 
To benchmark the processing of the modified multi-product algorithm, we simulate a general Rabi-type Hamiltonian of a coupled 
nuclear spin and electron spin system. 
To the best of our knowledge, this work reports the first multi-qubit controlled-unitary gate in integrated photonics, 
and we reprogram the device to implement the first Hamiltonian dynamics simulation in LCU circuits.

\begin{figure*}[ht!]
\centering 
\includegraphics[width=0.91\linewidth]{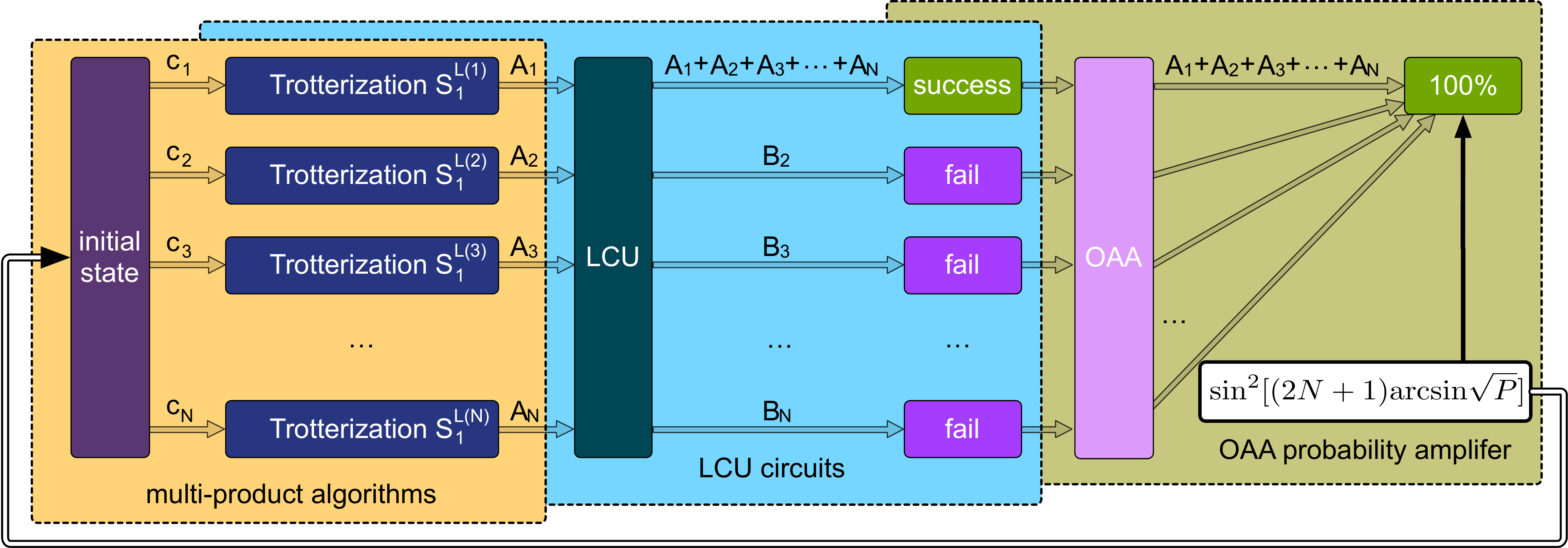}
\caption{\textbf{Architecture of quantum dynamics simulation using the modified multi-product algorithms with OAA in linearly-combined unitary circuits.} 
Multi-product algorithms (yellow boxed) can linearly combine multiple low-order Trotterizations ($\mathcal{S}_1^l$) to improve the simulation accuracy of quantum dynamics.
The LCU circuits (blue boxed) represent the quantum hardware that can implement the multi-product algorithms, and the LCU can be physically realized by a sequence of quantum controlled-unitary operations. 
There are many failing outcomes in the conventional LCU algorithms, resulting in a low probability of success. We adopt the OAA method (green boxed) to amplify the success probability of the algorithm.
This architecture allows enhanced multi-product algorithms  for quantum  simulation of dynamic evolution, with an improved accuracy and a high  success probability. By properly choosing the coefficients $ c_i$, iteration numbers $L(i)$ and $N$ so that $\sin^2[(2N+1)\arcsin\sqrt{P}]=1$, we can amplify the success probability up to unit.}
\label{fig:illustrate}
\end{figure*}

\vspace{-2mm}

\subsection*{Multi-product algorithms in LCU circuits}
\vspace{-2mm}

\noindent 
We first  overview the original multi-product algorithm~\supercite{childs2012hamiltonian}, and then discuss how to improve it with an assistance of the OAA algorithm~\supercite{berry2017exponential}. 
When considering the second-order Trotter approximation of a unitary  operator of $\exp(-i\sum_{j=1}^m \mathcal{H}_jt)$ given $\exp(-i\mathcal{H}_jt)$, we have: 
\begin{gather}
\begin{split}
    e^{-i\sum_{j=1}^m \mathcal{H}_jt} & = \mathcal{S}_1^l(t/l)+E_3 t^3/l^2 + E_5 t^5/l^4 +\cdots,\\
    \mathcal{S}_1(t) & = \prod_{j=1}^me^{-i\mathcal{H}_jt/2}\prod_{k=m}^1e^{-i\mathcal{H}_kt/2}
    \end{split}
\end{gather}
where $\mathcal{S}_1(t)$ is the second-order Trotter product, $E_3$ and $E_5$ are  high-order error terms  which are functions of $\mathcal{H}_j$,   $t$ is the evolution time of the quantum system, and $l$ is the number of iterations.   
Simply increasing the $l$, one can reduce the high-order errors but cannot get rid of them. The key idea of the multi-product algorithm is to properly combine multiple low-order Trotterizations linearly, so that the low-order error terms can be cancelled.  
For example, by properly choosing the coefficients of $ c_n = {n^2}/({n^2-m^2})$ and $ c_m={m^2}/({m^2-n^2})$,  and linearly combining two second-order Trotterizations with $n$ and $m$ iterations respectively, one can directly cancel the third-order error term $E_3$ and thus improve the approximation to the fifth-order error:  
\begin{gather}
    e^{-i\sum_{j=1}^m \mathcal{H}_jt} = c_n\mathcal{S}_1^n(t/n)+c_m\mathcal{S}_1^m(t/m) -E_5t^5/n^2m^2 + \mathcal{O}(t^7)
\end{gather}
That being said, the error of quantum dynamics simulation grows in $\mathcal{O}(t^5)$ by a linear  combination of two Trotterizations. 
Experimentally, this could be realised using a LCU of two second-order Trotterizations. 
In general, one  can repeat the procedure and reach an error of $\mathcal{O}(t^{2k+1})$ by the LCU of $k$ second-order Trotterizations with a proper choice of their coefficients~\supercite{chin2010multi}: 
\begin{gather}
	\begin{split}
		e^{-i\sum_{j=1}^m \mathcal{H}_jt} = 
		&\sum_{q=1}^kc_q\mathcal{S}_1^{L(q)}(t/L(q))+\\
		&(-1)^{k-1}E_{2k+1}t^{2k+1}\prod_{q=1}^{k }\frac{1}{L^2(q)} + \mathcal{O}(t^{2k+3}), 
	\end{split}
	\label{eq:scheme}
\end{gather}
where $c_q$ are chosen as $ \prod_{p\ne q, p=1}^k [{L^2(q)}/({L^2(q)-L^2(p)})]$, the $c_q$ satisfy $\sum_{q=1}^k c_q = 1$,  and $L(q)$ is the number of iteration for the $q$-th second-order formula (without loss of generality, we assume $L(q)$ is a monotonically increasing function).

Fig.\ref{fig:illustrate} illustrates  the scheme of multi-product algorithm, where $A_i$ represents  $c_iS_1^{L(i)}$ ($i=1,2,\cdots,N$) in Eq.(\ref{eq:scheme}), and its physical implementation in LCU circuits. 
The LCU circuits can be realized with a sequence of quantum controlled-unitary operations \supercite{gui2006general}, as shown in Fig.\ref{fig:device}\textbf{a}.  
Implementing the multi-product algorithm in the LCU circuits can significantly improve the Trotterization precision. 

While obtaining  the required linear combination of $B_1=\sum_{i=1}^{N}A_i$, the LCU circuits will also return the results of other linear combination terms ($B_2$,...,$B_N$) that should be in spam. These unwanted combinations are caused by the entanglement between ancillary and data register, and will result in a low success probability of the LCU circuits (see more details in Supplementary Information). 
The success probability drops quickly as the number of combined unitary gate grows. 
This is because $c_q$ and $c_{q+1}$ are different in sign, which makes the success probability of $
\text{Pr}=1/(\sum_{q=1}^k |c_q|)^2$ less than the unit.
To overcome this, one can design the function $L(q)$ to meet $\sum_{q=1}^k |c_q|\approx 1$ to improve the probability. 
However, there is a trade-off between high success probability ($\sum_{q=1}^k |c_q|\approx 1$) and low algorithm error ($\sum_{q=1}^k c_q = 1$) in original multi-product algorithm, making it impossible to target for both of them. It is this trade-off that reduces multi-product algorithm's potential in precision.

\begin{figure*}[ht!]
\centering 
\includegraphics[width=1\linewidth]{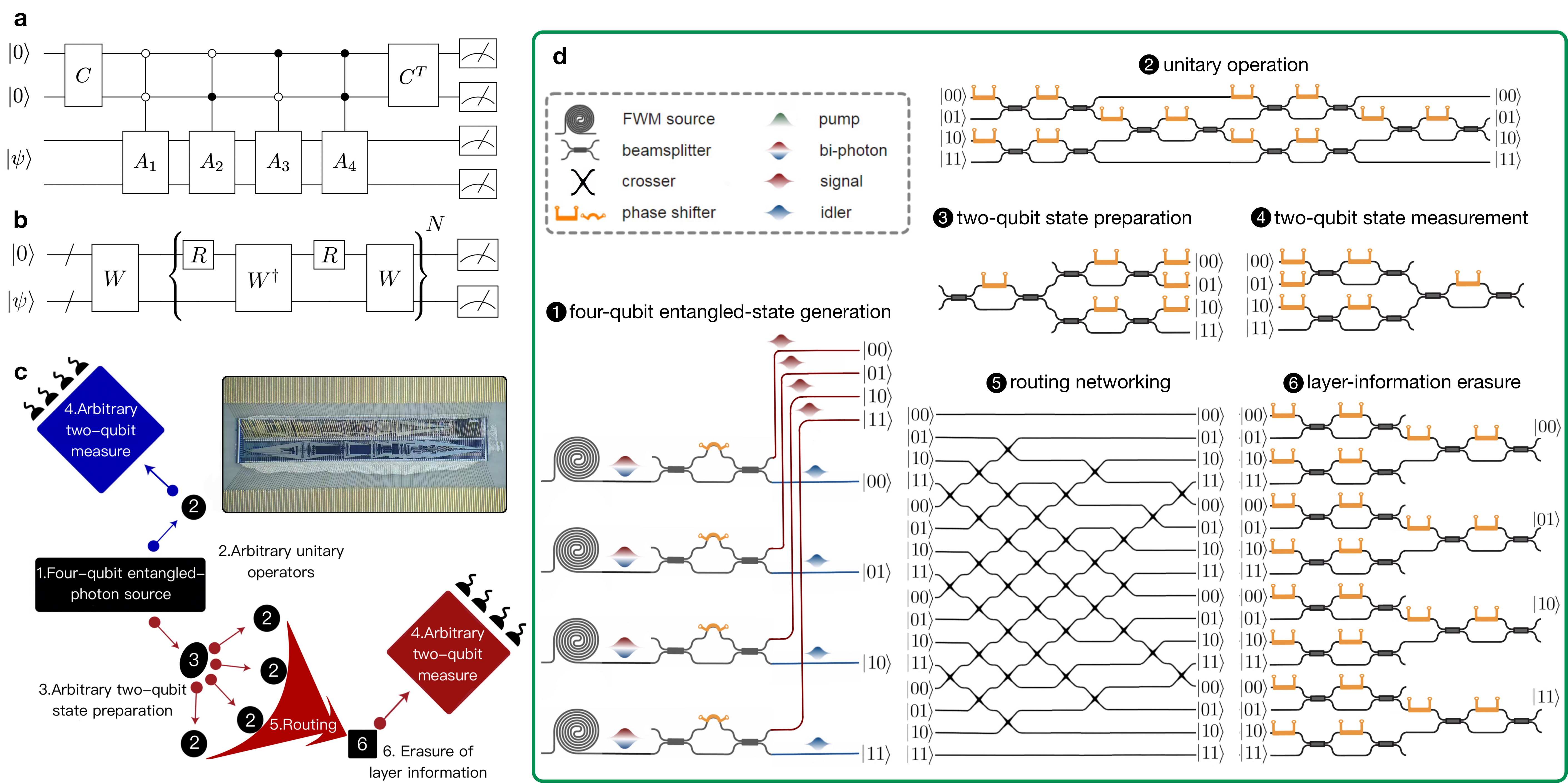}
\caption{\textbf{A programmable LCU quantum simulator in a silicon nanophotonic chip.}  
\textbf{a}, Quantum circuit for a four-qubit LCU circuit. 
The top two qubits refer to an ancillary register initialized in $\ket{00}$ state, and the bottom two qubits refer to a data register  prepared in $\ket{\psi}$ state. 
The LCU is enabled by a sequence of quantum controlled-unitaries. 
 \{$A_1,A_2,A_3,A_4$\} are arbitrary two-qubit unitaries, in which the target Hamiltonian is loaded in.  
For $C$ $(C^{T})$, its $i$-th element in the first column (row) 
is given by $({c_i/\sum_{q=1}^k |c_q|})^{1/2}$. 
\textbf{b},  OAA quantum circuit. The $W$ represents the LCU circuit described in \textbf{a}, and $R$ can flip the amplitude of basis $\ket{00}$ in the ancillary register. The bracketed circuit is repeated for $N$ times to amplify the successful probability of the LCU-based quantum algorithms, 
 when reading out outcomes by measuring the ancillary register in the $\ket{00}$ basis. 
 \textbf{c} and \textbf{d},  modularized scheme of the LCU quantum simulator. 
The simulator can implement the LCU circuit in \textbf{a}, and can be fully reprogrammed and reconfigured to implement the quantum algorithms. Inset: photograph of a LCU quantum simulator in a silicon chip.
 }
\label{fig:device}
\end{figure*}

\begin{figure}[ht!]
	\centering
	\includegraphics[width=.35\textwidth]{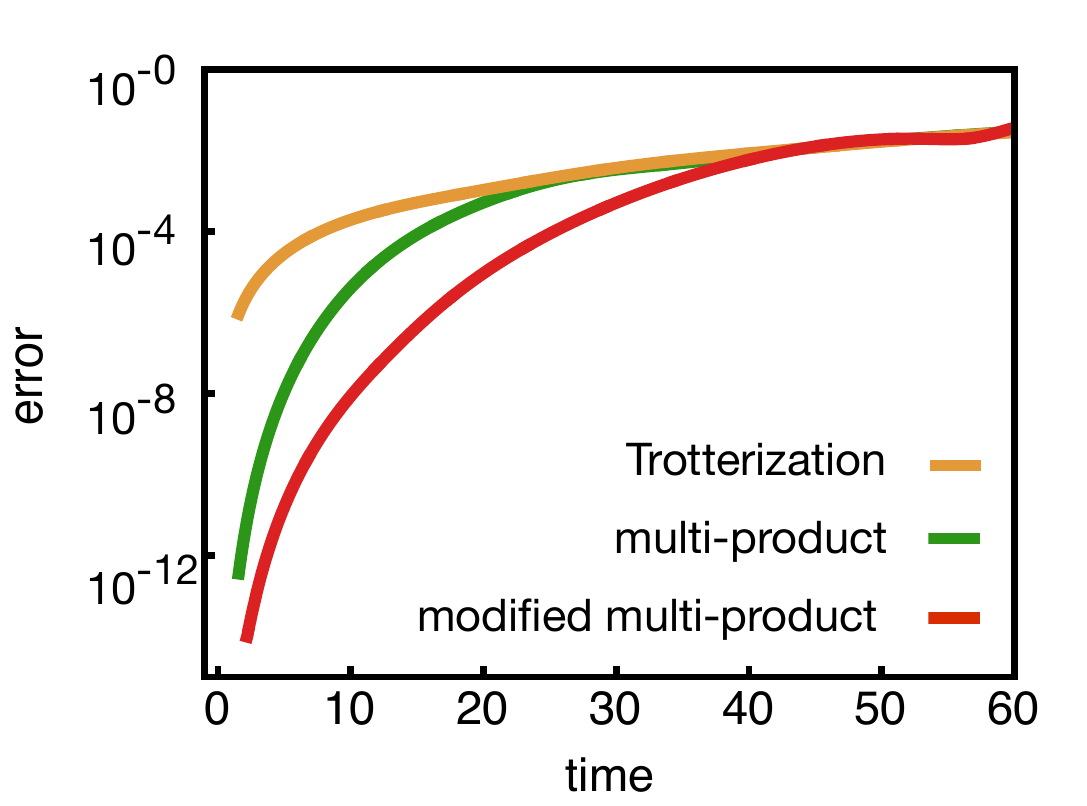}
	\caption{\textbf{Theoretical analysis of different Trotterization algorithms.} 	
		We compare the simulation error between the standard Trotter formula, multi-product and modified multi-product algorithms.
		The simulation error is defined as 
		$\left\| \ket{\psi_\text{0}} - \ket{\psi^\prime}\right\|$,  
		where $\ket{\psi_0}$  is the exact state and $\ket{\psi^\prime}$ refers to the output state of quantum algorithms. In the analysis, we consider a Hamiltonian (Eq.\ref{eq:H}) that  describes the interaction of an electron spin and a nuclear spin with initial state $\sqrt{0.3}\ket{00}+\sqrt{0.7}\ket{01}$. The iteration number ($l$) of Trotter formula, original and modified multi-product algorithms are $96, \{1,2,3,96\}$ and $\{4,8,16,32\}$, respectively. Note that the additional circuit depth and error of OAA have been taken into consideration.}
	\label{fig:comparison}
\end{figure}

We propose a modified multi-product algorithm, that is assisted by OAA, as shown in Fig.\ref{fig:illustrate}. 
The OAA algorithm can amplify the LCU success probability to near unit in a very similar way to Grover algorithm. The quantum circuit of OAA is plotted in Fig.\ref{fig:device}\textbf{b}. The OAA  is enabled by implementing a   $U_\text{OAA}$ = $ (-WRW^\dagger R)^N W$ circuit on the $\ket{0}\ket{\psi} $ state, where $W$ operator represents the LCU circuit,  $R$ is the amplitude flip operator defined as ($( I - 2\ket{0}\bra{0})\otimes I$), and  $N$ refers to the number of amplification iteration. 
When operating a single $W$ operator on the state, we have: 
\begin{gather}
    W\ket{0}\ket{\psi} = \sin\theta \ket{0}U\ket{\psi} + \cos\theta \ket{\Phi^\perp}, 
    \label{eq5}
\end{gather}
where the operator $U$ is the target operator of LCU, $\theta$ indicates the complex amplitude of expected state, and $ (\ket{0}\bra{0}\otimes I)\ket{\Phi^\perp} = 0$. If we project this state with projector $\ket{0}\bra{0}\otimes I$, it will turn into a typical LCU algorithm with success probability of $\sin^2(\theta)$. When operating the OAA operator on the state, 
given the $U$ is unitary, the OAA circuit will amplify $\theta$ and return the output state as: 
\begin{gather}
  U_\text{OAA}  \ket{0}\ket{\psi} = 
    \sin[(2N+1)\theta] \ket{0}U\ket{\psi} + \cos[(2N+1)\theta] \ket{\Phi^\perp}
\end{gather}
This process is very similar to Grover algorithm. By doing this, the OAA circuit can thus enhance the success probability from $P$ to $\sin^2[(2N+1)\arcsin(\sqrt{P})]$ (see more details about OAA's amplification effect in Supplementary Information). 
Obviously, OAA circuit requires larger circuit resource compared with LCU. In fact, it requires $\mathcal{O}((2N+1)\chi)$ resources, where $\mathcal{O}(\chi)$ is the resource for  LCU. 

One notice that, 
if the LCU circuit returns a unitary operator with a success probability of $25\%$, the OAA can amplify the probability to $100\%$ with $N=1$; if the LCU returns a non-unitary and its success probability is not exactly $25\%$, the OAA can also allow an efficient  amplification~\supercite{berry2015simulating}. 
To certificate the observation, we take an example of multi-product algorithm with $L(q)=2^q$. 
In this case, the $\sum_{q=1}^k|c_q|$ quickly converges to 1.969 as $k$ increases (e.g., $\sum_{q=1}^7|c_q|\approx 1.969$), which results in $P\approx25.79\%$. Having OAA, it allows a fast amplification of the success probability to be $99.96\%$. Note that any $L(q)\propto c^q$ yields the same $c_q$, therefore the success probability of $L(q) = a\cdot 2^q$ ($a$ is a positive integer that can be chosen arbitrarily) can also be amplified to nearly $100\%$.

Given by Eq.(\ref{eq:scheme}), we choose $L(q) = a\cdot2^q$, and amplify the output of LCU with OAA ($N=1$). This choice of $L(q)$ corresponds to LCU success probability $\approx 25\%$, OAA can thus amplify it to about $100\%$. Taking the non-unitary effect of LCU into account, we can both achieve an error of:
\begin{gather}
    \left\| \ket{\psi_\text{0}} - \ket{\psi^\prime}\right\|\le
    \left\|E_{2k+1}\right\|\frac{t^{2k+1}}{a^{2k} 2^{k^2+k-1}}
\end{gather}
which scales much slower than original multi-product algorithm as $k$ grows ($\ket{\psi_0}$ and $\ket{\psi^\prime}$ are the exact state and the output state of algorithms, respectively, $k$ is the number of combined unitary gates and $\left\|\cdot\right\|$ denotes the L2 norm for vector and spectral norm for matrix) and a nearly deterministic success probability. See more details about the derivation in Supplementary Information. 

In order to demonstrate modified multi-product algorithm's improvement in precision compared with original multi-product algorithm and Trotterization in a more intuitive way, we plot the corresponding simulation error in Fig.\ref{fig:comparison}. We choose $L(q)=2\cdot 2^q, q\in\{1,2,3,4\}$ for modified multi-product algorithm, $L(q)=1,2,3,96$ for original multi-product algorithm and $l=96$ for Trotterization to make them have similar circuit depths. 
This precision improvement will be more significant when we combine more unitary operators in algorithm.

\vspace{-2mm}
\subsection*{A programmable LCU-based photonic quantum simulator.}
\vspace{-2mm}

\noindent
Figures~\ref{fig:device}\textbf{c} and \ref{fig:device}\textbf{d} illustrate  the schematic of a programmable four-qubit silicon-photonic quantum simulator for the implementation
, and a photograph for the simulator is shown in the inset of Fig.\ref{fig:device}\textbf{c}. 
We implement the on-chip mapping of quantum dit (qudit) states\supercite{chi2022programmable} into multiple quantum bit (qubit) states. This allows us to implement high-fidelity and arbitrary entangling operations between the qubits, which are key to the implementation of LCU-based quantum algorithms. 

We  now discuss how to implement  the LCU circuit in the quantum simulator.  
As shown in Fig.\ref{fig:device}\textbf{a}, the LCU circuit is enabled by a sequence of programmable CCU gates.  
The photonic quantum simulator in Fig.\ref{fig:device}\textbf{d} integrates several key functional modules, including arbitrary local single-qubit preparation, quantum controlled-controlled unitary operation (CCU), and local projective measurement. The quantum simulator monolithically integrates 451 optical components, and it was fabricated on the Silicon-on-Insulator platform using complementary metal-oxide-semiconductor (CMOS) processes. Fabrication details are provided in Supplementary Information. 

We map quantum dit (qudit) states into multiple quantum bit (qubit) states. 
For example, in our photonic chip, we have the following mapping: 
\begin{equation}
	\left\{
	\begin{aligned}
		&\ket{0}_\text{qudit}\leftrightarrow\ket{00}_\text{qubit},~~
		\ket{1}_\text{qudit}\leftrightarrow \ket{01}_\text{qubit},\\
		&\ket{2}_\text{qudit}\leftrightarrow \ket{10}_\text{qubit},~~
		\ket{3}_\text{qudit}\leftrightarrow \ket{11}_\text{qubit}.~~
		\label{eq:mapping}
	\end{aligned}
	\right.
\end{equation}
In this work, the notation of $\ket{0}$, $ \ket{1}$ refers to the qubit representation. 

As shown in  Fig.\ref{fig:device}\textbf{d}, a pair of single photons are created in integrated spontaneous four-wave mixing (SFWM) sources. 
By coherently pumping four sources and by wavelength demultiplexing, it returns a four-dimensional Bell state~\supercite{Wang16D}.   
Using the mapping in Eq.({\ref{eq:mapping}}), we rewrite it as a four-qubit entangled state: 
\begin{equation}
	\left(\ket{0000}+\ket{0101}+\ket{1010}+\ket{1111}\right)/2, 
\end{equation}
where \{$\ket{0}, \ket{1}$\} represent the logical basis of each qubit. 
The basic idea is to coherently translate the entanglement in the source module to the CCU entangling gates \supercite{Zhou2011,Wang:QHL,Santagati2018}.

Firstly, we apply the space expansion procedure to create a six qubit entangled state: 
\begin{equation}
	\frac{1}{2}\left(\ket{00}\ket{00,\psi}+\ket{01}\ket{01,\psi}+\ket{10}\ket{10,\psi}+\ket{11}\ket{11,\psi}\right), 
\end{equation}
where the first two qubits represent the ancillary-register states, 
and the two-qubit state $\ket{\psi}$ represents the data-register and it can be arbitrarily initialized by module 3. The $\ket{ij,\psi}$ ($i,j=0,1$) denotes the state $\ket{\psi}$ encoded in the $(i,j)$-th layer qubits (will be erased eventually) that are entangled with the ancillary-register qubits. 

Secondly, we implement local operations $\{A_{1},A_{2},A_{3},A_{4}\}$  on the $\ket{\psi}$, where $A$ represents two-qubit unitary realised by  module 2.
We obtain a state of: 
\begin{equation}
	\begin{split}
		\frac{1}{2}(&\ket{00}A_{1}\ket{00,\psi}+\ket{01}A_{2}\ket{01,\psi}+\\
		&\ket{10}A_{3}\ket{10,\psi}+\ket{11}A_{4}\ket{11,\psi}). 
	\end{split}
\end{equation} 

Thirdly,
we implement projective measurement $M$ (module 4) on the qubits of ancillary-register, and post-select its $\ket{00}\bra{00}$ outcomes. 
This process 
projects the data-register into an output of 
\begin{equation}
	m_1A_{1}\ket{00,\psi}+ m_2A_{2}\ket{01,\psi}+ m_3A_{3}\ket{10,\psi}+ m_4A_{4}\ket{11,\psi}, 
\end{equation}
where \{$m_1,m_2,m_3,m_4$\} correspond to the first column of the $C$ operator in Fig.\ref{fig:device}\textbf{a}.

\begin{figure}[t]
	\centering 
	\includegraphics[width=0.34\textwidth]{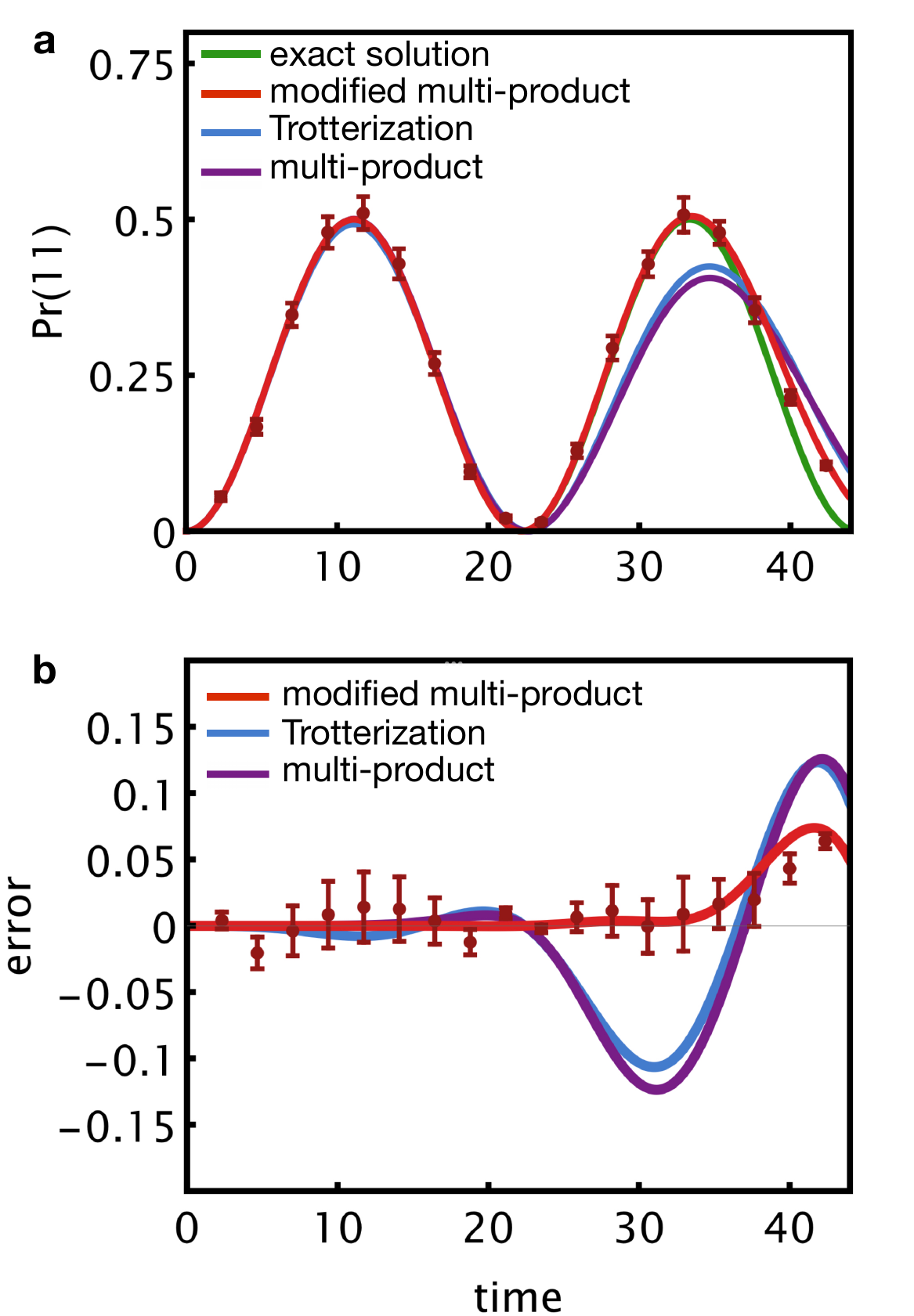}
	\caption{\textbf{Emulating quantum dynamics in a LCU quantum simulator.} \textbf{a}, A comparison of standard Trotterization, multi-product and the modified multi-product algorithm, in term of their simulation accuracy. In the data register, an initial state $\ket{01}$ is input into the simulator, and dynamics for the state of $\ket{11}$ is measured. 
\textbf{b}, Simulated error for different algorithms. Error is defined as the discrepancy of exact solution and computed results. 
Quantum dynamic evolution of four different states. 
Points denote experimental data, and lines refer to theoretical predictions. 
Both theoretical prediction and experiment show that the modified multi-product algorithm can achieve higher precision than the Trotterization and multi-product algorithm, when $t\le 30$. 
Error bars ($\pm 1 \sigma$) are estimated from the photon Poissonian statistics.}
	\label{fig:LCUResult}
\end{figure}

Finally, the coherent compression modules (modules 5,6) implements the $C^{\prime}$ operator, as well as the projector of $\ket{00}_{p}\bra{00}$ acting on the layer qubits 
$\ket{ij}$ in $\ket{ij, \psi}$. Then the output state is given as: 
\begin{gather}
	\begin{split}
		&m_1\ket{00}_p\bra{00}C^{\prime} A_{1}\ket{00,\psi} + m_2\ket{00}_p\bra{00}C^{\prime} A_{2}\ket{01,\psi} +\\ &m_3\ket{00}_p\bra{00}C^{\prime} A_{3}\ket{10,\psi} + m_4\ket{00}_p\bra{00}C^{\prime} A_{4}\ket{11,\psi}, 
	\end{split}
\end{gather} 
where the unitaries $\{C^{\prime},A_{1},A_{2},A_{3},A_{4}\}$ and complex numbers $\{m_1,m_2,m_3,m_4\}$ can be arbitrarily controlled and reprogrammed.  
The above post-selection only takes account of the elements in the first row of the matrix representation of gate $C^{\prime}$. Similar procedures can be applied to other projective post-selections. 
\begin{figure*}[t]
	\centering 
	\includegraphics[width=0.95\linewidth]{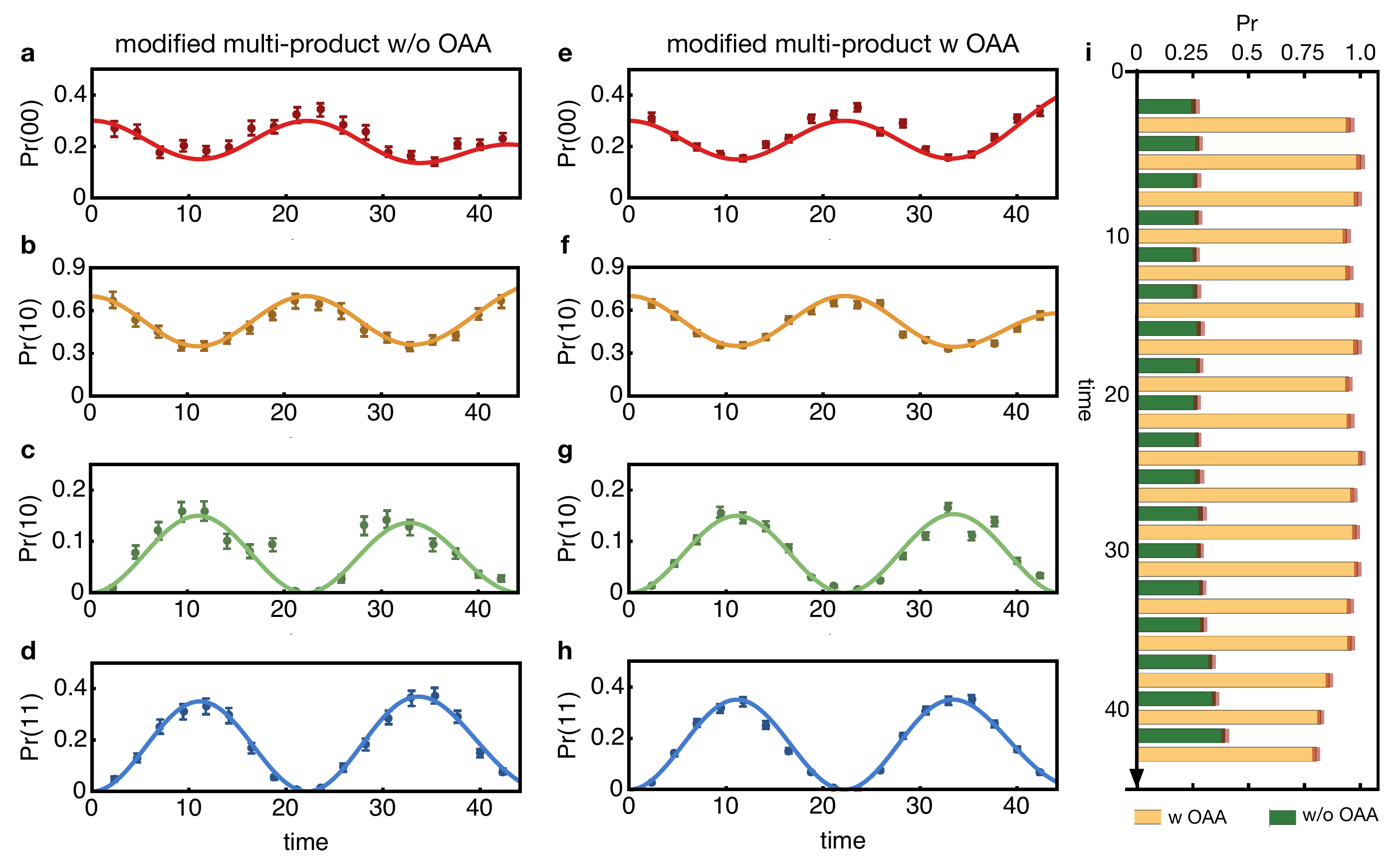}
	\caption{\textbf{Experimental results of modified multi-product algorithms with and without OAA.} 
\textbf{a-d}, Quantum dynamics of four states \{$\ket{00}, \ket{01}, \ket{10}, \ket{11}$\}, simulated using the modified multi-product algorithm without the OAA. \textbf{e-h}, Quantum dynamics of four  states \{$\ket{00}, \ket{01}, \ket{10}, \ket{11}$\}, simulated using the modified multi-product algorithm with the OAA. 
	\textbf{i}, Measured success probability for the modified multi-product algorithm with OAA ($99.79\%$ for $t<30$) and without OAA ($26.33\%$ for $t<30$). The success probability of the modified multi-product algorithms thus can be amplified by OAA. 
	Points denote experimental data, and lines refer to their theoretical predictions. Error bars ($\pm 1 \sigma$) are estimated from the photon Poissonian statistics. 
	}
	\label{fig:OAAResult}
\end{figure*}
Assuming the first row of $C^\prime$ is $\{m_1^\prime,m_2^\prime,m_3^\prime,m_4^\prime\}$ (arbitrary complex numbers), we obtain the output state of data-register as:  
\begin{gather}
	\begin{split}
		\ket{\Psi}_\text{data}   =& (m_1m_1^\prime A_{1}+ m_2m_2^\prime A_{2}+ 
		m_3m_3^\prime A_{3}+ m_4m_4^\prime A_{4})\ket{00,\psi}\\
		=& (m_1m_1^\prime A_1 + m_2m_2^\prime A_2+m_3m_3^\prime A_3+m_4m_4^\prime A_4)\ket{\psi},\label{eq4}
	\end{split}
\end{gather}
where we have removed the layer index  for clarity. 
Though $m_i$ and $m_i^\prime$ can be chosen arbitrarily, 
the combined coefficients $  m_im_i^\prime = {c_i}/{\sum_{i=1}^{4} |c_i|} $ can be optimized to improve the success probability of the algorithm.  
When $m_i=m_i^\prime$, 
the success probability of LCU becomes the largest (see proof in Supplementary Information). For this reason, we choose $C^\prime$ to be $C^T$ (see Fig.\ref{fig:device}\textbf{a}), so that the elements in first column of $C$ and first row of $C^T$ are the same.

\vspace{-2mm}
\subsection*{Experimental benchmarking of multi-product algorithms}
\vspace{-2mm}
\noindent
We benchmark the modified multi-product algorithms in our LCU quantum simulator to simulate the dynamics of a Hamiltonian describing the Rabi-type interactions of an electron spin and a nuclear spin in the rotating frame of external electromagnetic field frequency. The Hamiltonian $\mathcal{H}$ is written as: 
\begin{equation}
    \mathcal{H}=\frac{\Omega}{2}\sigma_x^e+\frac{\delta}{2}\sigma_z^e+\ket{1}_e\bra{1}\otimes (\textbf{E}_1\ket{0}_n\bra{0}+\textbf{E}_2\ket{1}_n\bra{1}), 
    \label{eq:H}
\end{equation}
where $\sigma_x$ and $\sigma_ z$ are Pauli matrices, subscript and superscript of $n$ and $e$ represent the nuclear and electron spins, respectively, $\ket{0}_n$ and $\ket{1}_n$ represent the ground state and first excited state with eigenenergies of $E_1$ and $E_2$, and $\Omega$ is the frequency of the external electromagnetic field. 
  The Hamiltonian is defined by the following parameters in experiment, $E_1=0.3, E_2=0.7, \Omega=0.2, \delta=0.5$. 
We use the LCU quantum simulator to emulate the unitary $U(t)=\text{exp}(-i\mathcal{H}t)$. 
The operators \{$A_1,A_2,A_3,A_4$\}  are chosen as \{$S_1^2(t/2)$, $S_1^4(t/4)$, $S_1^8(t/8)$, $S_1^{16}(t/16)$\}, respectively. The $S_1(t)$ is second-order Trotter product of 
$ e^{-i\mathcal{H}_1t/2}e^{-i\mathcal{H}_2t}e^{-i\mathcal{H}_1t/2} $, where 
   $ \mathcal{H}_1$ is the first two terms in Eq.(\ref{eq:H}), $ \mathcal{H}_2$ is the last term.

Fig.\ref{fig:LCUResult}\textbf{a} shows experimental probability distributions for the $\ket{11}$ state measured in the LCU simulator, emulating the dynamic evolution of a $\ket{01}$ initial state governed by the Hamiltonian. 
We compare the results of modified multi-product algorithm (combination of $S_1^2(t/2)$, $S_1^4(t/4)$, $S_1^8(t/8)$ and $S_1^{16}(t/16)$) with standard Trotterization (e.g., $S_1^{16}(t/16)$) and original multi-product algorithm (combination of $S_1(t)$, $S_1^2(t/2)$, $S_1^3(t/3)$ and $S_1^{16}(t/16)$). 
It shows that, for $t\le 30$, the modified multi-product algorithm can approximate the exact solution more precisely than the other two algorithms. Their discrepancies to the exact result are plotted in Fig.\ref{fig:LCUResult}\textbf{b}. 
For $t\textgreater 30$, both original and modified multi-product show deviation from the exact solution due to high-order error and the non-unitary effect. 
 
Figure \ref{fig:OAAResult}(\textbf{a-h}) reports experimental dynamics for the four logical states measured in the quantum simulator, when the initial state  is chosen as $\sqrt{0.3}\ket{00}+\sqrt{0.7}\ket{01}$, i.e., $\sqrt{0.3}\ket{0}_e\ket{0}_n+\sqrt{0.7}\ket{0}_e\ket{1}_n$. Both modified multi-product algorithm before OAA (Figs.\ref{fig:OAAResult}(a-d)) and after OAA (Figs.\ref{fig:OAAResult}(e-h)) are implemented. Note  that the OAA circuits were  complied and implemented by iteratively input into the quantum simulator. 
Experimental results well agree with their theoretical predictions, and a mean of classical statistical fidelities of  $F_c=0.996\pm0.005$ and $F_c=0.996\pm0.004$ are obtained, respectively, for the modified multi-product algorithm with and without OAA. 
The $F_c$ is defined as $(\sum_i\sqrt{p_iq_i})^2$, where $p_i$ and $q_i$ are  measured distribution and theoretical distribution respectively. 
Moreover, 
we measured the success probability of algorithms, as shown in Fig~\ref{fig:OAAResult}\textbf{i}. 
Experimental results show that, when $t\le 30$, the success probability $P$ is about $26.33\%$ in the modified multi-product algorithm before OAA step. 
We then implement OAA and obtain the amplified probability. 
As the initial $Pr$ from LCU was approximated to $26.33\%$, it requires only one step of amplification, i.e. $N=1$. 
The OAA circuit in Fig.\ref{fig:device}\textbf{b} can be simplified to $WRW^\dagger RW$, and compiled   in the simulator.  
By implementing the OAA, we amplify the success probability to $99.79\%$ at $t\le30$, approaching a deterministic probability. 
All measured probabilities are consistent with  theoretical predictions. 
When $t>30$, due to the non-unitary LCU, the high-level success probability however cannot be retained. 

Overall we have proposed and demonstrated a modified multi-product algorithm. 
By combing the multi-product and OAA algorithms, it can improve the simulation precision and relax the success probability constraints.
The algorithm is certificated experimentally in a small-scale silicon-photonic quantum simulator. 
Moreover, the linearly combined operations are not necessarily unitary (mostly, they are non-unitary), 
implying that LCU can be used to simulate the dynamics of non-Hermitian systems. 
The quantum algorithms and quantum harware could be scaled up to provide an efficient solution to simulate the dynamics of physical and molecular systems.

\subsection*{Acknowledgements} 
\vspace{-2mm}
\noindent We acknowledge support from the Innovation Program for Quantum Science and Technology (no. 2021ZD0301500), the National Key R$\&$D Program of China (no 2019YFA0308702), the Natural Science Foundation of China (no 61975001), Beijing Natural Science Foundation (Z190005), and Key R$\&$D Program of Guangdong Province (2018B030329001). 

\vspace{-2mm}
\subsection*{Author contributions} \vspace{-2mm}
\noindent Y.Y. and Y.C. equally contributed  to this work. 
Y.Y., Y.C. and J.H. implemented the experiment. 
Y.Y., Y.C. and C.Z. provided theoretical analysis. Q.G. and J.W. managed the project. 
Y.Y., Y.C., and J.W. wrote the manuscript. 

\vspace{-2mm}
\subsection*{Competing interests:} \vspace{-2mm}
\noindent The authors declare no competing interests.

\vspace{-2mm}
\subsection*{Data availability}\vspace{-2mm}
\noindent The data that support the plots within this paper and other findings of this study are available from the corresponding author upon reasonable request.

\printbibliography[notkeyword=SI]

\clearpage
\pagenumbering{arabic}
\setcounter{page}{1}
\onecolumn
\section*{\centering\fontsize{14}{14}\selectfont 
Supplementary Information: \\
Simulating Hamiltonian dynamics in a programmable photonic quantum processor using linear combinations of unitary operations
}
\begin{center}
\begin{minipage}{1\textwidth}
\begin{center}
\begin{spacing}{1.5}

{Yue Yu$^{1,\star}$, Yulin Chi$^{1,\star}$, Chonghao Zhai$^{1}$, Jieshan Huang$^{1}$, Qihuang Gong$^{1,2,3,4,5}$,  Jianwei Wang$^{1,2,3,4,5,\dagger}$}

\textit{\textrm{
\textsuperscript{1} State Key Laboratory for Mesoscopic Physics, School of Physics, Peking University, Beijing, 100871, China
\\\textsuperscript{2} Frontiers Science Center for Nano-optoelectronics \& Collaborative Innovation Center of Quantum Matter, Peking University, Beijing, 100871, China 
\\\textsuperscript{3} Collaborative Innovation Center of Extreme Optics, Shanxi University, Taiyuan 030006, Shanxi, China
\\\textsuperscript{4} Peking University Yangtze Delta Institute of Optoelectronics, Nantong 226010, Jiangsu, China.
\\\textsuperscript{5} Hefei National Laboratory, Hefei 230088, China
}}
\end{spacing}
\end{center}
\end{minipage}
\end{center}

\setcounter{figure}{0}
\makeatletter 
\renewcommand{\thefigure}{S\@arabic\c@figure}
\makeatother

\setcounter{table}{0}
\makeatletter 
\renewcommand{\thetable}{S\@arabic\c@table}
\makeatother

\begin{spacing}{1.5}

\vspace{-3mm}
\noindent \textbf{Device fabrication.}
The integrated photonic four-qubit quantum simulator used to implement the linear combinations of unitaries for Hamiltonian simulation is designed and fabricated on the silicon nanophotonics platform. The  devices are fabricated by the standard CMOS (complementary metal oxide semiconductor) process. First, spin a layer of photoresist on an 8-inch SOI (silicon on insulator) wafer with 220 nm thick top silicon and 3 \micro m  thick buried oxide. 248 nm deep ultraviolet (DUV) lithography is used to define the pattern of the silicon layer on the photoresist where single-mode waveguides are designed with a width of 500nm (widened to 4 microns in some areas to reduce transmission losses). The pattern is transferred from the photoresist layer to the silicon layer by double inductively coupled plasma (ICP) etching process to form waveguide and grating couplers. A deeply etched waveguide with an etching depth of 220 nm is used for the SFWM photon source, beam splitter(multimode interferometer) and phase shifter. A shallow etched waveguide with an etching depth of 70 nm is used for waveguide crossovers and grating couplers. A 1\micro m thick SiO\textsubscript{2} layer is deposited on the SOI wafer by plasma-enhanced chemical vapor deposition (PECVD) as an isolation layer between the waveguide and the metal heater. Then, a 10 nm thick Ti adhesive layer, a 20 nm thick TiN barrier layer, an 800 nm thick AlCu layer and a 20 nm thick TiN antireflection layer were deposited by physical vapor deposition (PVD), in which AlCu was used as a conductive layer and patterned by DUV lithography and etching process to form electrodes. Then, a 50 nm thick TiN layer for the thermooptical phase shifter was deposited by PECVD and patterned by DUV lithography and etching process. Finally, 1 \micro m thick SiO\textsubscript{2} is deposited as the top coating. The bonding pad opening process is then carried out to remove the SiO\textsubscript{2} and TiN after DUV lithography in specific areas designed for wire bonding.

\vspace{3mm}

\noindent {\textbf{Experimental setup.} 
A tunable continuous-wave laser central at a wavelength of 1550.12 nm was amplified to 100 mW by an erbium doped fiber amplifier (EDFA). Then the amplified laser is filtered by array waveguide gratings (AWG) wavelength-division multiplexing (WDM) to remove noise from the amplifier, and then coupled to the chip through on-chip grating couplers. In order to improve the coupling efficiency of the grating coupler, we use a fiber polatization controller to adjust the polarization of the laser before the chip to achieve the highest coupling efficiency, which is about -4 dB each port.
The entangled photon pairs are generated in on-chip sources through the spontaneous four-wave mixing (SFWM) effect of silicon from the pump laser. We chose 1545.32nm and 1554.94nm as the wavelength of idler and signal photons, and photons of other wavelengths and the residual pump light are filtered out by off-chip AWG with a bandwidth of 0.9 nm and extinction ratio $\textgreater90$ dB. 
The entangled photons were detected by superconducting nanowire single-photon detectors (SNSPD) with 85\% efficiency in average and coincident counting measurements were conducted using a multichannel time interval analyzer (TIA). The resulting two-photon coincidence count rate is at the order of kHz.
The integrated photonic chip is wire bonded on a printed circuit board (PCB) and the PCB is fixed with thermal epoxy on a copper constant temperature terrace to maintain the modulation temperature of the thermo-optic phase modulator, providing stable temperature conditions and suppressing thermal crosstalks.
At the design of practical devices, we use curved waveguides to compensate the length of different structures to ensure coherence. At the same time, we also designed several 1\% couplers at some positions not marked in the figure to calibrate the phase modulators on the chip. They all use the same on-chip grating couplers to couple with the off-chip fibers.
}

\vspace{3mm}

\noindent {\textbf{Linear combination of unitary operations.}} For readers unfamiliar with LCU, we provide a brief introduction to it here. The LCU algorithm utilizes the entanglement between two quantum systems: N-dimensional ancillary register $\ket{\cdot}_a$ and M-dimensional data register $\ket{\cdot}_d$, and non-deterministically generates the linear combination of $N$ unitary gates $\sum_{i=1}^N c_iA_i$ ($c_i$ is coefficient and $A_i$ is unitary gate) that acts on data register.

Firstly, a product state, $\ket{0}_a\ket{\psi}_d$, is generated as the input state of LCU circuit. Since there is no entanglement between ancillary and data register in this quantum state, it is relatively easy to generate.

Secondly, we implement a N-qubit gate $C$ on the ancillary register:
\begin{gather}
	C =
	\begin{bmatrix}
		\sqrt{\frac{c_1}{\sum_{i=1}^N|c_i|}}&*&\cdots&*\\
		\sqrt{\frac{c_2}{\sum_{i=1}^N|c_i|}}&*&\cdots&*\\
		\vdots&\vdots&\ddots&\vdots\\
		\sqrt{\frac{c_N}{\sum_{i=1}^N|c_i|}}&*&\cdots&*\\
	\end{bmatrix}
\end{gather}
It can be seen from the above expression that the j-th element in first column of  $C$'s matrix representation is $\sqrt{c_j/\sum_{i=1}^N|c_i|}$. We don't care about elements in other columns since only elements in the first column affect the evolution of $\ket{0}_a$. After this step, the quantum state becomes
\begin{gather}
	\sqrt{\frac{1}{\sum_{i=1}^N|c_i|}}\sum_{j=1}^N \sqrt{c_j}\ket{j-1}_a\ket{\psi}_d
\end{gather}

Thirdly, we implement N C(N)U(M) gates (M-qubit gate controlled by N control qubits) $A_i$, $i=1,2,\cdots,N$. These gates are controlled by ancillary register and act on data register. This step projects the quantum state into an output of
\begin{gather}
	\sqrt{\frac{1}{\sum_{i=1}^N|c_i|}}\sum_{j=1}^N \sqrt{c_j}\ket{j-1}_aA_j\ket{\psi}_d
\end{gather}

Finally, we implement N-qubit gate $C^T$ on ancillary register. The matrix representation of $C^T$ is 
\begin{gather}
	C^T =
	\begin{bmatrix}
		\sqrt{\frac{c_1}{\sum_{i=1}^N|c_i|}}&\sqrt{\frac{c_2}{\sum_{i=1}^N|c_i|}}&\cdots&\sqrt{\frac{c_N}{\sum_{i=1}^N|c_i|}}\\
		*&*&\cdots&*\\
		\vdots&\vdots&\ddots&\vdots\\
		*&*&\cdots&*\\
	\end{bmatrix}
\end{gather}
This will lead to following output:
\begin{gather}
	\frac{1}{\sum_{i=1}^N |c_i|}\ket{0}_a\sum_{j=1}^N c_jA_j\ket{\psi}_d + \ket{1}_a\cdots + \ket{2}_a\cdots + \cdots
\end{gather}
It can be seen from the above expression that each ancillary register state $\ket{j}_a$ corresponds to a linear combination of $A_i$, but only the combination after $\ket{0}_a$ is what we want. Therefore, we post-select the ancillary register with projector $\ket{0}_a\bra{0}$. If we assume $\left\|\sum_{j=1}^Nc_jA_j\ket{\psi}_d\right\|=1$ (this assumption does not always hold, for example, when $\sum_{j=1}^Nc_jA_j$ is not unitary), then the post-selection probability is $\left(1/\sum_{i=1}^N|c_i|\right)^2$. Generally, this possibility drops quickly as $N$ increases.
\vspace{3mm}

\noindent {\textbf{Oblivious amplitude amplification.}
	Oblivious Amplitude Amplification~\supercite{berry2017exponential} (OAA) is an algorithm  proposed to solve the low success probability problem of LCU. It is very similar to Grover algorithm and the post-selection probability can be amplified by multiple iterations of the LCU circuit. To describe its main structure, we first define the following operators:
	\begin{gather}
		\begin{split}
			W\ket{\Psi} &= \sin(\theta)\ket{\Phi}+\cos(\theta)\ket{\Phi^\perp}\label{oaa1}\\
			P &= \ket{0}\bra{0}\otimes I\\
			R &= (1-2\ket{0}\bra{0})\otimes I
		\end{split}		
	\end{gather}
	where $\ket{\Psi}, \ket{\Phi}$ and $\ket{\Phi^\perp}$ have following relationships:
	\begin{align}
		\ket{\Psi} &= \ket{0}\ket{\psi}\\ 
		\ket{\Phi} &= \ket{0}U\ket{\psi}\\
		P &\ket{\Phi^\perp} = 0
	\end{align}
	$\ket{\psi}$ is arbitrary quantum state and $U$ is unitary operator. Obviously, $W$ can be used to describe a general LCU circuit and $P$ is the corresponding post-selection projector. Based on the above definitions, OAA tells us that the following equation holds:
	\begin{gather}
		(-WRW^\dagger R)^NW\mathinner{|\Psi\rangle} = \sin[(2N+1)\theta]\mathinner{|\Phi\rangle}+\cos[(2N+1)\theta]\mathinner{|\Phi^\perp\rangle}
	\end{gather}	
	Therefore, we can amplify the post-selection probability from $\sin^2(\theta)$ to $\sin^2((2N+1)\theta)$ with OAA.
	
	The verification process is shown below:
	
	Firstly, we define state $\ket{\Psi^\perp}$ as
	\begin{gather}
		W\ket{\Psi^\perp} = \cos(\theta)\ket{\Phi} - \sin(\theta)\ket{\Phi^\perp}\label{oaa2}
	\end{gather}
	It can be easily proved that the state $\ket{\Psi^\perp}$ always exists.
	
	The following lemma will be very useful in the verification of OAA:
	\newtheorem{lemma}{\bf Lemma}
	\begin{lemma}
		$ P \ket{\Psi^\perp} = 0$ when $\theta \ne 0$
		\label{oaalemma1}
	\end{lemma}
	\begin{proof}
		We define operator $Q:=(\bra{ 0}\otimes I)W^\dagger P W(\ket{0}\otimes I)$. Then $\forall\ket{\psi}$, we have:
		\begin{gather}
			\begin{split}
				\mathinner{\langle\psi|Q|\psi\rangle}     
				&= \left\| P W(\mathinner{|0\rangle}\otimes I)\mathinner{|\psi\rangle}\right\|\\
				&= \left\| P W\mathinner{|\Psi\rangle}\right\|\\     
				&= \left\| P\left(\sin(\theta)\mathinner{|\Phi\rangle} +\cos(\theta)\mathinner{|\Phi^\perp\rangle}\right)\right\|\\     
				&= \sin^2(\theta)
			\end{split}
		\end{gather}
		Since this holds for all $\ket{\psi}$, we have $Q = \sin^2(\theta)I$, $Q\ket{\psi} =\sin^2(\theta)\ket{\psi}$.
		
		From another direction, we can calculate $Q\ket{\psi}$ directly:
		\begin{gather}
			\begin{split} 
				Q\mathinner{|\psi\rangle} 
				&= (\mathinner{\langle 0|}\otimes I)W^\dagger P W(\mathinner{|0\rangle}\otimes I)\mathinner{|\psi\rangle}\\ 
				&= (\mathinner{\langle 0|}\otimes I)W^\dagger P W\mathinner{|\Psi\rangle}\\
				&=\sin(\theta) (\mathinner{\langle 0|}\otimes I)W^\dagger\mathinner{|\Phi\rangle}\\ 
			\end{split}\\
		\end{gather}
		Therefore
		\begin{gather}
			\sin^2(\theta)\ket{\psi} = \sin(\theta) (\mathinner{\langle 0|}\otimes I)W^\dagger\mathinner{|\Phi\rangle}\label{oaa3}
		\end{gather}
		
		From formula~\ref{oaa1} and~\ref{oaa2}, we have:
		\begin{gather}
			W^\dagger\ket{\Phi} = \sin(\theta)\ket{\Psi} + \cos(\theta)\ket{\Psi^\perp}
		\end{gather}
		Bring it back to formula~\ref{oaa3}
		\begin{gather}
			\begin{split} 
				\sin^2(\theta)\mathinner{|\psi\rangle}  
				&= \sin(\theta)(\mathinner{\langle 0|}\otimes I)W^\dagger\mathinner{|\Phi\rangle}\\ 
				&= \sin(\theta)(\mathinner{\langle 0|}\otimes I)(\sin(\theta)\mathinner{|\Psi\rangle}+\cos(\theta)\mathinner{|\Psi^\perp\rangle})\\ 
				&= \sin(\theta)\mathinner{|\psi\rangle}+\sin(\theta)\cos(\theta)(\mathinner{\langle 0|}\otimes I)\mathinner{|\Psi^\perp\rangle} \end{split}\\
		\end{gather}
		Therefore
		\begin{gather}
			(\bra{0}\otimes I )\ket{\Psi^\perp} =  P\ket{\Psi^\perp} =0
		\end{gather}
	\end{proof}

	Based on \textbf{Lemma}~\ref{oaalemma1}, we can prove the amplification effect of OAA:
	
	\newtheorem{thm}{\bf Theorem}
	\begin{thm}
		$(-WRW^\dagger R)^NW\ket{\Psi} = \sin[(2N+1)\theta]\ket{\Phi} + \cos[(2t+1)\theta]\ket{\Phi^\perp}$
		\label{oaathm1}
	\end{thm}
	\begin{proof}
		We define $S := -WRW^\dagger R$. Based on the previous discussion, we have:
		\begin{align}
			W\mathinner{|\Psi\rangle} &= \sin(\theta)\mathinner{|\Phi\rangle}+\cos(\theta)\mathinner{|\Phi^\perp\rangle}\\ 
			W\mathinner{|\Psi^\perp\rangle} &= \cos(\theta)\mathinner{|\Phi\rangle}-\sin(\theta)\mathinner{|\Phi^\perp\rangle}\\ 
			W^\dagger\mathinner{|\Phi\rangle} &= \sin(\theta)\mathinner{|\Psi\rangle}+\cos(\theta)\mathinner{|\Psi^\perp\rangle}\\ 
			W^\dagger\mathinner{|\Phi^\perp\rangle} &= \cos(\theta)\mathinner{|\Psi\rangle}-\sin(\theta)\mathinner{|\Psi^\perp\rangle}
		\end{align}
		Therefore
		\begin{gather}
			SW\mathinner{|\Psi\rangle} = \sin(\theta)S\mathinner{|\Phi\rangle}+\cos(\theta)S\mathinner{|\Phi^\perp\rangle}\label{eq6}
		\end{gather}
		
		According to \textbf{Lemma}~\ref{oaalemma1}, we can write down following relationships:
		\begin{align}
			R\ket{\Phi} &= \ket{\Phi}\\
			R\ket{\Phi^\perp} &= -\ket{\Phi^\perp}\\
			R\ket{\Psi} &= \ket{\Psi}\\
			R\ket{\Psi^\perp} &= -\ket{\Psi^\perp}
		\end{align}
		
		Therefore
		\begin{gather}
			\begin{split} 
				S\mathinner{|\Phi\rangle} 
				&= -W RW^\dagger R\mathinner{|\Phi\rangle}= -WRW^\dagger \mathinner{|\Phi\rangle}\\ 
				&= -WR\left(\sin(\theta)\mathinner{|\Psi\rangle}+\cos(\theta)\mathinner{|\Psi^\perp\rangle}\right)\\ 
				&= W\left(-\sin(\theta)\mathinner{|\Psi\rangle}+\cos(\theta)\mathinner{|\Psi^\perp\rangle}\right)\\ 
				&= (\cos^2(\theta)- \sin^2(\theta))\mathinner{|\Phi\rangle}-2\sin(\theta)\cos(\theta)\mathinner{|\Phi^\perp\rangle}\\ 
				&= \cos(2\theta)\mathinner{|\Phi\rangle}-\sin(2\theta)\mathinner{|\Phi^\perp\rangle}
			\end{split}
		\end{gather}
	
		\begin{gather}
			\begin{split} 
				S\mathinner{|\Phi^\perp\rangle} 
				&= -W RW^\dagger R\mathinner{|\Phi^\perp\rangle}= WRW^\dagger \mathinner{|\Phi^\perp\rangle}
				\\ &=WR \left(\cos(\theta)\mathinner{|\Psi\rangle}-\sin(\theta)\mathinner{|\Psi^\perp\rangle}\right)
				\\ &=W \left(\cos(\theta)\mathinner{|\Psi\rangle}+\sin(\theta)\mathinner{|\Psi^\perp\rangle}\right)
				\\ &=2\sin(\theta)\cos(\theta)\mathinner{|\Phi\rangle}+(\cos^2(\theta) - \sin^2(\theta))\mathinner{|\Phi^\perp\rangle}
				\\ &= \sin(2\theta)\mathinner{|\Phi\rangle} + \cos(2\theta)\mathinner{|\Psi\Phi^\perp\rangle} 
			\end{split}
		\end{gather}
		Take these back to formula~\ref{eq6}, we have
		\begin{gather}
			SW\mathinner{|\Psi\rangle} = \sin(3\theta)\mathinner{|\Phi\rangle}+\cos(3\theta)\mathinner{|\Phi^\perp\rangle}
		\end{gather}
		Similarly, we can further prove that
		\begin{gather}
			S^NW\mathinner{|\Psi\rangle} = \sin[(2N+1)\theta]\mathinner{|\Phi\rangle}+\cos[(2N+1)\theta]\mathinner{|\Phi^\perp\rangle}
		\end{gather}
	\end{proof}
}
\vspace{3mm}
\vspace{3mm}

\noindent{
\textbf{Error analysis of OAA.} OAA is originally used to amplify the LCU with unitary target operator. When the target operator implemented by LCU is unitary and the success probability of LCU is exactly $25\%$, OAA can amplify the success probability into $100\%$ without any additional error. However, this condition does not always hold.
}

We define $P := \ket{0}\bra{0}\otimes I$, $R := I-2P$ and denote the circuit of LCU as $W$, then we have 
\begin{gather}
    PW\ket{0}\ket{\psi} = s\ket{0}U\ket{\psi}
\end{gather}
Where $U$ is the target operator implemented by LCU, $s$ is the complex amplitude of target operator and success probability is $|s|^2$. We then denote the circuit of OAA as $A=WRW^\dagger R W$, therefore
\begin{gather}
    PA\ket{0}\ket{\psi} = (3PW - 4WPW^\dagger PW)\ket{0}\ket{\psi}
\end{gather}

Following the conclusion proposed by \textit{Berry et al.}\supercite{berry2015simulating}, we obtain:
\begin{gather}
    PA\ket{0}\ket{\psi} = (3sU-4s^3UU^\dagger U)\ket{\psi}
\end{gather}

When $\left\|U U^\dagger -I\right\|=\mathcal{O}(\delta)$ and $s = 0.5+\Delta$, we can bound the error of OAA:
\begin{gather}
    \left\|\frac{PA\ket{\psi}}{\left\|PA\ket{\psi}\right\|} - \frac{PW\ket{\psi}}{\left\|PW\ket{\psi}\right\|}\right\| = \left\|\frac{PA\ket{\psi}}{\left\|PA\ket{\psi}\right\|} - U\ket{\psi}\right\| \approx (\frac{1}{2}+3\Delta)\mathcal{O}(\delta)
\end{gather}

So when the success probability of LCU is close to $25\%$ (for example, probability deviation $\le 5\%$), it will not affect the error of OAA too much. If we follow the choice $L(q) = a\cdot 2^q$, then $\Delta\approx 8\times 10^{-3}$ and error $\approx 0.52\mathcal{O}(\delta)$. We need to emphasize here that $\mathcal{O}(\delta)$ represents the non-unitary effect of LCU, and the error of LCU is not necessarily equal to $\mathcal{O}(\delta)$.  However, we will prove in the next section that the error of LCU is larger than $\mathcal{O}(\delta)$ in unitary simulation task, which can help us bound the overall error of modified algorithm.

\vspace{3mm}
\noindent{\textbf{Error analysis of modified algorithm}}. There are two source of error in our modified algorithm: error introduced by OAA amplification and intrinsic error of non-unitary LCU.

Firstly, we focus on the amplification error of OAA. As we have discussed previously, under proper choice of linear combination coefficients in LCU, the amplification error is
\begin{gather}
	\left\|E_{oaa}\right\| \approx\frac{\mathcal{O}(\delta)}{2}
\end{gather}
where $\mathcal{O}(\delta) = \left\|UU^\dagger - I\right\|$, $U$ is the target operator of LCU. Notice that the goal of modified algorithm is to simulate a unitary operator $U_{exact}$, therefore we denote $U = U_{exact} + E_{lcu}$, ($E_{lcu}$ is not necessarily unitary), then we have
\begin{gather}
	\begin{split}
		\left\|UU^\dagger - I\right\|
		&= \left\|UU_{exact}^\dagger + E_{lcu}U^\dagger + UE_{lcu}^\dagger + E_{lcu}E_{lcu}^\dagger - I\right\|\\
		&= \left\|E_{lcu}U^\dagger + UE_{lcu}^\dagger + E_{lcu}E_{lcu}^\dagger\right\|\\
		&\le \left\|E_{lcu}U^\dagger\right\| + \left\|UE_{lcu}^\dagger\right\| + \left\|E_{lcu}E_{lcu}^\dagger \right\|\\
		&\le 2\left\| E_{lcu}\right\| + \left\| E_{lcu}\right\|^2\\
		&\approx 2\left\| E_{lcu}\right\|
	\end{split}
\end{gather}
Therefore
\begin{gather}
	\left\| E_{lcu}\right\| \ge\frac{\mathcal{O}(\delta)}{2}\approx \left\|E_{oaa}\right\|
\end{gather}
So the overall error of the modified algorithm $\left\|E_{total}\right\|\le\left\| E_{lcu}\right\| +\left\| E_{oaa}\right\| \le 2\left\|E_{lcu}\right\|$, therefore the error of the modified algorithm is bounded by error of LCU when combination coefficients are chosen properly.

Next we discuss the error of LCU. The error of multi-product is\supercite{chin2010multi}
\begin{gather}
	e^{-i\sum_{j=1}^m \mathcal{H}_jt} = \sum_{q=1}^kc_q\mathcal{S}_1^{L(q)}(t/L(q))+ (-1)^{k-1}E_{2k+1}t^{2k+1}\prod_{q=1}^{k }\frac{1}{L^2(q)}+\mathcal{O}(t^{2k+3})
\end{gather}
If we choose $L(q) = a\cdot 2^q$, then the error is
\begin{gather}
	\left\|E_{lcu}\right\| = \left\|E_{2k+1}\right\|\frac{t^{2k+1}}{a^{2k} 2^{(1+k)k}}
\end{gather}
And then then total error of modified algorithm is bounded as below:
\begin{gather}
	\left\| E_{total}\right\| \le \left\|E_{2k+1}\right\|\frac{t^{2k+1}}{a^{2k} 2^{k^2+k-1}}
\end{gather}

\vspace{3mm}
\noindent{
	\textbf{Comparison between multi-product algorithm and modified version.}}
 In this section, we will discuss the advantages and drawbacks of multi-product algorithm and its modified version. The advantage of modified version is its high success probability and precision. Here we will compare the two version in details. In original version~\supercite{childs2012hamiltonian} , $L(q)$ is chosen to be
	\begin{gather}
		L(q) = 
		\left
		\{
		\begin{array}{ll}
			q&  1\le q<k\\
			e^{\gamma k}& q = k 
		\end{array}
		\right.
	\end{gather}
	and the error of multi-product algorithm is
	\begin{gather}
		\left\| E_{2k+1}\right\| t^{2k+1}\frac{e^{-2\gamma k}}{[(k-1)!]^2}
	\end{gather}
	For modified version, we choose $L(q)$ the nearest integer to $e^{\gamma k}\cdot 2^{q-k} / 3$, so the modified version has the same circuit depth with original version (Trotter products in LCU are commutative). In this case, the error of modified version is
	\begin{gather}
		\left\| E_{2k+1}\right\| t^{2k+1}9^k 2^{k^2-k+1}e^{-2\gamma k^2}
	\end{gather}
	It's easy to prove that the second error is smaller when $k$ is large.
	
	In contrast, although the original version multi-product algorithm's precision is lower, its error grows much slower then modified version. We attribute this to Trotter formula's property. The error we calculate only works when the Taylor expansion's truncation error is negligible. When the truncation error can't be ignored, the error of modified version will grows dramatically, while the Trotter formula's error will grow much slower. This is because in original multi-product algorithm, the Trotter product with the largest iteration number contributes most of the simulation result (because $\sum_{q=1}^{k-1}|c_q|\ll |c_k|$), so it has similar property as Trotter formula.

\vspace{3mm}

\noindent {\textbf{Selection of $\bf m_i$ and $\bf m_i^\prime$.}
When the non-unitary effect of LCU is insignificant, $m_i$ and $m_i^\prime$ determine the outcome and success probability of LCU (in the following text, the success probability refers to the probability at this stage). Assuming that the expected operator of LCU is $\sum_{i=1}^k c_iA_i$, then $m_i$ and $m_i^\prime$ should satisfy
\begin{gather}
    \frac{m_im_i^\prime}{m_jm_j^\prime} = \frac{c_i}{c_j}\quad i,j=1,2,\cdots, k\label{lcu1}
\end{gather}

Meanwhile, $\{m_1, m_2, \cdots, m_k\}$ and $\{m_1^\prime, m_2^\prime, \cdots, m_k^\prime\}$ is the first column of $C$ and the first row of $C^\prime$, respectively. Since $C$ and $C^\prime$ are unitary operator, we have
\begin{gather}
    \sum_{i=1}^k |m_i|^2 = \sum_{i=1}^k|m_i^\prime|^2 = 1\label{lcu2}
\end{gather}

The success probability of LCU is
\begin{gather}
    P = |\frac{m_im_i^\prime}{c_i}|^2 = \frac{|m_i|^2|m_i^\prime|^2}{|c_i|^2}\quad\forall i\in \{1,2,\cdots,k\}
\end{gather}
therefore
\begin{gather}
    \frac{|c_i|^2}{|m_i|^2} = \frac{|m_i^\prime|^2}{P}\\
    \sum_{i=1}^k\frac{|c_i|^2}{|m_i|^2} = \sum_{i=1}^k\frac{|m_i^\prime|^2}{P} = \frac{1}{P}
\end{gather}
The last equality is because of formula~\ref{lcu2}. Since $\sum_{j=1}^k |m_j|^2=1$, we have
\begin{gather}
    \sum_{j=1}^k |m_j|^2\sum_{i=1}^k\frac{|c_i|^2}{|m_i|^2} = \frac{1}{P}
\end{gather}
Therefore
\begin{gather}
\begin{split}
    \frac{1}{P}
    &=\sum_{i=1}^k |c_i|^2 + \sum_{i,j=1,i\ne j}^{k}|c_i|^2\frac{|m_j|^2}{|m_i|^2} \\
    &\ge \sum_{i=1}^k |c_i|^2 + \sum_{i,j=1,i\ne j}^{k}|c_i| |c_j|\\
    &= (\sum_{i=1}^k |c_i|)^2
\end{split}
\end{gather}
The inequality is because of $|c_i|^2\frac{|m_j|^2}{|m_i|^2} + |c_j|^2\frac{|m_i|^2}{|m_j|^2} \ge 2|c_i| |c_j|$. Thus
\begin{gather}
    P\le \frac{1}{(\sum_{i=1}^k |c_i|)^2}
\end{gather}

We can achieve this limit with $m_i=m_i^\prime = \sqrt{c_i/\sum_{i=1}^{k} |c_i|}$, therefore the success probability under this choice is the largest.
}
\end{spacing}

\end{document}